\documentclass[%
 reprint,
superscriptaddress,
 amsmath,amssymb,
 aps,
]{revtex4-2}

\usepackage{graphicx}
\usepackage{dcolumn}
\usepackage{bm}
\usepackage{amsmath,amssymb,amsthm,mathtools}
\usepackage{physics}
\usepackage{comment}
\usepackage{mathtools}
\usepackage{xcolor}
\usepackage{mathalpha}

\newtheorem{theorem}{Theorem}

\newtheorem{lemma}{Lemma}

\begin{document}

\preprint{APS/123-QED}

\title{Non-Gaussianity in cat codes: global incompatibility and local geometric alignment with the magic resource}

\author{Yuwei Zhu}
 \affiliation{Okinawa Institute of Science and Technology Graduate University, Onna-son, Okinawa 904-0495, Japan
}
\author{Seungbeom Chin}
\affiliation{Okinawa Institute of Science and Technology Graduate University, Onna-son, Okinawa 904-0495, Japan
}
\author{William J. Munro}
\affiliation{Okinawa Institute of Science and Technology Graduate University, Onna-son, Okinawa 904-0495, Japan
}
\author{Kae Nemoto}
\affiliation{Okinawa Institute of Science and Technology Graduate University, Onna-son, Okinawa 904-0495, Japan
}
\affiliation{National Institute of Informatics, 2-1-2 Hitotsubashi, Chiyoda-ku, Tokyo 101-8430, Japan}

\begin{abstract}
Non-Gaussianity is an essential resource for genuine quantum advantages in continuous-variable quantum systems and is regarded as a counterpart of the magic resource in discrete-variable systems.
Recently, an exact correspondence between non-Gaussianity (NG) and the magic resource was identified within the Gottesman--Kitaev--Preskill (GKP) encoding framework. Whether such a relation persists beyond GKP encoding, however, remains unclear.
Here, we address this question in the cat-code setting. By comparing the Wigner logarithmic negativity (WLN) and a magic measure defined from a phase-operator basis, we analyze the resource geometry of non-degenerate $d$-peaked cat states. We find that cat codes do not inherit the global value-preserving GKP magic--NG equivalence, but their asymptotic WLN geometry allows a constructive local alignment with the magic-measure geometry. 
Specifically, under a distinguished SU($d$) asymptotic cat code, we establish a sector-dependent alignment between WLN level sets and magic-measure level sets based on their intrinsic local geometries. 
These results identify both the global incompatibility and the locally constructive relation between magic and non-Gaussian resources in cat codes, suggesting a geometry-based approach to resource correspondences beyond the GKP framework.

\end{abstract}

\maketitle


\section{Introduction}
Continuous-variable (CV) quantum systems, characterized by infinite-dimensional Hilbert spaces, constitute a major framework for quantum information processing, and provide a natural setting for tasks such as universal CV quantum computation \cite{LloydBraunstein1999}. A central insight in this framework is that Gaussian ingredients alone are fundamentally limited: broad classes of Gaussian processes remain efficiently classically simulable \cite{Bartlett2002,Mari2012,Veitch2013,RahimiKeshari2016,Walschaers2021}, and therefore cannot by themselves account for genuine quantum advantage in such tasks. This places non-Gaussianity (NG) at the center of CV quantum information, across platforms ranging from quantum optics to superconducting bosonic modes~\cite{Konno2024,sivak2023}. Correspondingly, resource-theoretic quantifiers such as Wigner logarithmic negativity (WLN) and related measures have become useful computable indicators of this resource \cite{Kenfack2004,Genoni2010,Takagi2018,Albarelli2018}.

In discrete-variable (DV) quantum computation, an analogous role is played by non-stabilizerness, or the quantum magic resource. Clifford-stabilizer computation is classically tractable, whereas non-stabilizer resources are necessary for universal quantum computation \cite{Gottesman1998,BravyiKitaev2005,Veitch2014}. This has given rise to a well-developed resource-theoretic understanding of the magic resource. Since non-Gaussian and magic resources play closely parallel roles in their respective computational settings, the connection between them has long been conceptually anticipated~\cite{Gross2006,Hahn2022,Bu2023}. Major recent advances have established~\cite{Hahn2022,Feng2024,Hahn2025}, through Gottesman--Kitaev--Preskill (GKP) encoding~\cite{GKP2001}, a quantitative correspondence between the two resources. More specifically, in an appropriate phase-space operator basis, the discrete-variable magic quasidistribution coincides exactly with that of the encoded GKP state Wigner function, so that the associated $l_p$-norm-based resource measures agree on the two sides~\cite{Hahn2025}.

Despite the intriguing and remarkable exactness of the GKP correspondence, it appears to rely on the special phase-space structure of the Gottesman--Kitaev--Preskill (GKP) encoding, including lattice periodicity and stabilizer structure.
This raises the question of \emph{whether a weaker, geometry-based relation between magic and non-Gaussian resources can persist beyond the GKP framework.} One natural direction is to ask whether such a relation can emerge in other bosonic codes.
Cat codes provide an appealing setting for this question, as they are finite-energy coherent-state encodings and constitute a well-developed bosonic-code platform. Schr\"odinger-cat superpositions were among the earliest bosonic encodings proposed for error-protected logical qubits, and have since developed into a broad framework for efficient bosonic quantum computation, especially because of their strong noise bias under photon loss and the possibility of bias-preserving operations \cite{Cochrane1999,Ralph2003,Mirrahimi2014,Bergmann2016,MIRRAHIMI2016,Guillaud2019,Puri2020,CAI2021}. Meanwhile, cat states are highly accessible non-Gaussian bosonic states and have been extensively studied both theoretically and experimentally \cite{Vlastakis2013,Leghtas2015,Touzard2018,Grimm2020,Lescanne2020}. Despite the operational role of cat-state non-Gaussianity in bosonic quantum information processing, its geometric structure in phase space remains unclear. In particular, it is not evident whether cat encodings admit a meaningful counterpart of the magic--NG correspondence.

Our geometrical analysis demonstrates that the global level-set structures of the magic measure and WLN in the cat--code setting are incompatible; therefore, \emph{the form of any magic–NG relation is encoding-dependent.}
On the other hand, in the large-separation regime, the NG quantified by the WLN exhibits a phase-rotation-invariant local geometric structure. This allows us to define a distinguished SU($d$) asymptotic cat code (ACC) under which the level sets of the magic measure and WLN are locally aligned. 
Moreover, with respect to the local contour-deviation metric used here, the physical cat-code construction reaches comparable geometric tolerance at moderate cat amplitudes, whereas the approximate GKP family considered here requires much larger scale-separation parameters. Thus, we can construct an ACC that provides a distinct and physically accessible route to probing local geometric connections between the magic resource and NG.

This work is organized as follows. In Sec.~\ref{sec:structure}, we analyze the non-Gaussian structure of $d$-peaked cat states in the large-separation regime and identify the associated geometry of the WLN. In Sec.~\ref{sec:compatible}, we develop a sector-eigenbasis description of the qudit magic measure and show that its global level-set structure is intrinsically incompatible with the WLN geometry of standard cat encodings. We then use the local invariants shared by the two resources to construct, for arbitrary dimension $d$, an ACC whose encoded physical WLN locally aligns with the level-set geometry of the logical magic measure. In Sec.~\ref{sec:realizations}, we compare the finite-parameter requirements for realizing this local alignment with those of approximate GKP codewords. Finally, Sec.~\ref{sec:conclusions} presents our conclusions and outlook.

\section{Asymptotic WLN geometry of qudit cat states}\label{sec:structure}
In this section, we analyze the non-Gaussian structure of cat states with multiple coherent-state superpositions. We consider a generalized cat state defined as a superposition of $d$ coherent states,
\begin{equation}
    \ket{\psi_d}
    =
    \frac{1}{\sqrt{\mathcal N_d}}
    \sum_{i=0}^{d-1} c_i \ket{\alpha_i},
    \label{eq:quditCat}
\end{equation}
where $\{\alpha_i\}_{i=0}^{d-1}$ are arbitrary coherent-state amplitudes and
$\{c_i\}_{i=0}^{d-1}$ is a normalized complex coefficient vector satisfying $\sum_{i=0}^{d-1}|c_i|^2=1$. Without loss of generality, one may choose $c_0$ as a non-negative real number. Since coherent states are not mutually orthogonal, the normalization factor is
\begin{equation}
    \mathcal N_d
    =
    \sum_{i,j=0}^{d-1}
    c_i^* c_j
    \exp\!\left[
        -\frac{|\alpha_i|^2+|\alpha_j|^2}{2}
        + \alpha_i^*\alpha_j
    \right].
    \label{eq:normalized}
\end{equation}
In the large-separation regime, namely $|\alpha_i-\alpha_j|\gg 1$ for all $i\neq j$, we have $\mathcal N_d \approx \sum_i |c_i|^2 = 1$. We examine the corresponding Wigner function of \eqref{eq:quditCat}. Adopting quadrature variables $(x,p)$ with $[\hat x,\hat p]=i$, and writing $\alpha_i=(x_i+i p_i)/{\sqrt2}$, the Wigner function of $\ket{\psi_d}$ can be decomposed into diagonal Gaussian
contributions and off-diagonal interference terms as
\begin{equation}
\begin{split}
    &W_{\ket{\psi_d}}(x,p)
    =\\
    &\frac{1}{\mathcal N_d}
    \left[
        \sum_{i=0}^{d-1}|c_i|^2\,G_i(x,p)
        +
        2\sum_{0\le i<j\le d-1}|c_i c_j|\,I_{\alpha_i,\alpha_j}(x,p)
    \right],
    \end{split}
    \label{eq:wigner}
\end{equation}
where
\begin{equation}
\begin{split}
    G_{\alpha_i}(x,p)
    &=\frac{1}{\pi}
    \exp\!\left[
        -(x-x_i)^2-(p-p_i)^2
    \right],\\
    I_{\alpha_i,\alpha_j}(x,p)
    &=
    G_{(\alpha_i+\alpha_j)/2}\cos\!\left[\Delta_{ij}(x,p)+\phi_{ij}\right].
    \label{eq:dCatCoeff}
\end{split}
\end{equation}
Here $G_{\alpha_i}$ is the Gaussian peak centered at $(x_i,p_i)$ and $I_{\alpha_i,\alpha_j}(x,p)$ is the pairwise interference contribution. The factor $G_{(\alpha_i+\alpha_j)/{2}}$ is the Gaussian peak centered at the midpoint between the two coherent-state centers, and $\Delta_{ij}(x,p)
    =
    (p_i-p_j)x-(x_i-x_j)p
    +(x_i p_j-p_i x_j)/2,$ where $\phi_{ij}=\arg(c_i)-\arg(c_j)$ is the relative phase between the coefficients $c_i$ and $c_j$. The Wigner function of a $d$-component cat state consists of two distinct parts:
(i) $d$ Gaussian peaks associated with the classical phase-space locations of the coherent components, and (ii) $\binom{d}{2}$ pairwise interference fringes. Each interference term is modulated by a Gaussian envelope centered at the midpoint between $\alpha_i$ and $\alpha_j$, while its oscillation wavevector is set by the phase-space separation between the two coherent states. In particular, larger separations produce more rapidly oscillating fringes. Since each diagonal Gaussian contribution is nonnegative, any negativity of the Wigner function must originate from the interference terms. This will play a central role in the following resource-theoretic analysis based on the WLN.

\subsection{WLN and level-set structures in the large-separation regime}
To evaluate the WLN of the $d$-peaked cat state in Eq.~\eqref{eq:quditCat}, we consider a \textit{non-degenerate large-separation limit} of the coherent-state constellation. More precisely, let \(\{\beta_i\}_{i=0}^{d-1}\) be a fixed reference constellation. We assume that it is non-degenerate in the sense that
\begin{equation}
\beta_i+\beta_j \neq \beta_k+\beta_l,
\qquad
\text{unless } \{i,j\}=\{k,l\},
\label{eq:nondegenerate_constellation}
\end{equation}
for any unordered pairs $i,j$ and $k,l$, including the diagonal cases $i=j$ and $k=l$. We then define the scaled coherent amplitudes by
\begin{equation}
\alpha_i=R\beta_i,
\;R\to\infty.
\end{equation}
The resulting constellation $\{\alpha_i\}_{i=0}^{d-1}$ is said to be in the non-degenerate large-separation limit. For any \(i\neq j\), it satisfies
\begin{equation}
|\alpha_i-\alpha_j|
\rightarrow \infty
\end{equation}
and the condition in Eq.~\eqref{eq:nondegenerate_constellation} ensures that the centers of all Gaussian envelopes appearing in the Wigner-function decomposition are mutually distinct. 
Under the scaling \(\alpha_i=R\beta_i\), the separations between any two distinct envelope centers grow linearly with \(R\), whereas the Gaussian envelope widths remain \(\mathcal O(1)\). Therefore, in the large-separation limit, the overlaps between different diagonal and interference contributions are exponentially suppressed, while the coherent components also become asymptotically orthogonal. Consequently, the absolute-value integral entering the WLN decomposes asymptotically into a sum of individual localized contributions, leading to the closed-form WLN expressions derived below. For brevity, we will refer to this limit simply as the large-separation limit.

To quantify the non-Gaussianity of a continuous-variable state, we use the
Wigner logarithmic negativity (WLN)~\cite{Takagi2018,Albarelli2018}, defined by
\begin{equation}
    \mathrm{WLN}(\rho)
    =
    \log_2\!\left(\iint_{\mathbb{R}^2} |W_\rho(x,p)|\,dx\,dp\right).
\end{equation}
In the large-separation regime, each pairwise interference term in Eq.~\eqref{eq:dCatCoeff} becomes rapidly oscillatory under a Gaussian envelope. As a result, for every pair \((i,j)\), the interference term satisfies the asymptotic relation
\begin{equation}
\lim_{|\alpha_i-\alpha_j|\to\infty}
    \iint_{\mathbb{R}^2}
    |I_{\alpha_i,\alpha_j}(x,p)|\,dx\,dp
    =
    \frac{2}{\pi},
    \label{eq:limit}
\end{equation}
which shows that the absolute value of each interference contribution averages to \(2/\pi\) under the Gaussian envelope in the large-separation limit. See Appendix~\ref{app:limit} for more mathematical details. Meanwhile, the normalization factor
$\mathcal N_d$ approaches unity as all coherent-state overlaps are suppressed exponentially. These yield the non-Gaussianity of the $d$-peaked cat state in the asymptotic regime.

\begin{theorem}[Asymptotic WLN]
\label{thm:asymptotic_d_cat_WLN}
Consider the $d$-peaked cat state \(\ket{\psi_d}\) in Eq.~\eqref{eq:quditCat}, and assume that the coherent-state constellation is taken in the non-degenerate large-separation limit. Then the Wigner logarithmic negativity of $\ket{\psi_d}$ is 
\begin{equation}
    \mathrm{WLN}_\infty(\ket{\psi_d})
    =
    \log_2\!\left(
        1+\frac{4}{\pi}
        \sum_{0\le i<j\le d-1}|c_i c_j|
    \right),
    \label{eq:dWLN}
\end{equation}
where the subscript \(\infty\) denotes this limit.
\end{theorem}
Theorem~\ref{thm:asymptotic_d_cat_WLN} follows from the asymptotic decomposition of the absolute-value integral in Eq.~\eqref{eq:limit}. In this limit, different diagonal and pairwise interference envelopes are non-overlapping up to exponentially small corrections. Since each conjugate off-diagonal pair has already been combined into the real cosine fringe $2|c_i c_j|I_{\alpha_i,\alpha_j}$ in Eq.~\eqref{eq:wigner}, the absolute-value integral decomposes into independent localized contributions. Using $\sum_i |c_i|^2=1$, WLN in Eq.~\eqref{eq:dWLN} can equivalently be written as $\mathrm{WLN}_\infty(\ket{\psi_d})
=\log_2[1+2((\sum_{i=0}^{d-1}|c_i|)^2-1)/\pi]$. 

Given the expression of Eq.~\eqref{eq:dWLN}, we now identify several key properties in the asymptotic regime. These properties give rise to a rigid geometric structure of the WLN level sets, which underlies the code constructions developed later. First, Eq.~\eqref{eq:dWLN} immediately identifies the extremal non-Gaussianity of the cat states in Eq.~\eqref{eq:quditCat} in the asymptotic regime. 
The asymptotic WLN reaches its minimum value,
\(\mathrm{WLN}_{\infty}^{\min}(d)=0\), when the state is supported on a single coherent component. These \(d\) coherent-component states form singular points of the asymptotic WLN geometry. It attains its maximum value
\begin{equation}
        \mathrm{WLN}_{\infty}^{\max}(d)
    =
    \log_2\!\left(1+\frac{2}{\pi}(d-1)\right),
    \label{eq:WLNextremal}
\end{equation}
for $|c_0|=\cdots=|c_{d-1}|=1/\sqrt d$, when all amplitudes are equally distributed. This upper bound grows logarithmically with the number of coherent components $d$, as illustrated in Fig.~\ref{fig:maxWLN}.
\begin{figure}
    \centering
    \includegraphics[width=1\linewidth]{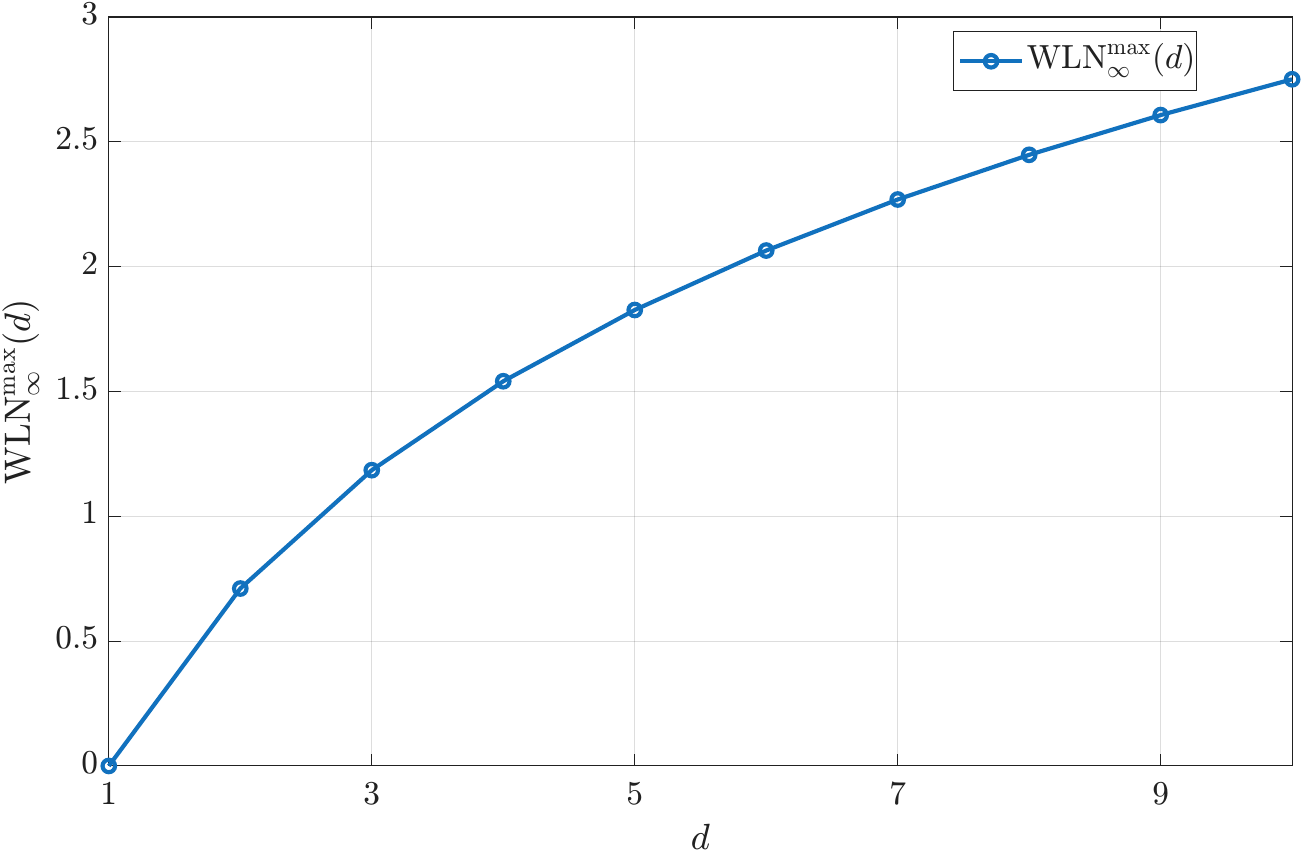}
    \caption{Asymptotic upper bound of the WLN for $d$-component cat states in the large-separation regime, $d=1,...,10$. The bound is attained for the equal-amplitude superposition and increases logarithmically with the number of coherent components $d$, according to Eq.~\eqref{eq:WLNextremal}.}
    \label{fig:maxWLN}
\end{figure}

Second, in the asymptotic regime, the WLN depends only on the amplitudes $|c_i|$ and is invariant under the relative phases of the coefficients. In other words, if we fix an asymptotic WLN value $\mathrm{WLN}_\infty(\ket{\psi_d})=C$ to describe the level sets of states with equal WLN, the constraints on the states are: (i) the amplitudes satisfy
\begin{equation}
    \sum_{i=0}^{d-1}|c_i|
    =
    \sqrt{1+\frac{\pi}{2}\bigl(2^C-1\bigr)};
    \label{eq:constWLNAmp}
\end{equation}
(ii) after quotienting out the global phase, the transformations
\begin{equation}
    c_i\mapsto e^{i\theta_i}c_i,\;i=1,\dots,d-1
\end{equation}
generate $d-1$ natural $U(1)$ phase directions along which $\mathrm{WLN}_\infty$ remains unchanged.
Thus, each asymptotic WLN level set is characterized by a fixed amplitude constraint together with an associated family of phase-rotation symmetries in state space. 

Third, the asymptotic WLN structure admits a natural encoding interpretation. In the large-separation limit, the coherent states in $\{|\alpha_i\rangle\}_{i=0}^{d-1}$ become mutually orthogonal and therefore define a valid $d$-dimensional code subspace. Within this subspace, one may consider logical basis changes generated by
\begin{equation}
    |{\mu}\rangle_L
    =
    \sum_{i=0}^{d-1} U_{\mu i}\,|\alpha_i\rangle,
    \qquad U\in \operatorname{SU}(d).
    \label{eq:sudInvariance}
\end{equation}
Eq.~\eqref{eq:dWLN} defines a coefficient-space geometry with respect to the asymptotically orthogonal coherent-state components. Choosing another logical basis represents this same coefficient-space geometry in rotated logical coordinates. In this sense, the intrinsic phase-rotation structure of the asymptotic WLN can be found in different logical bases. This basis freedom will be used below to construct a cat-code basis whose local WLN geometry aligns with that of the logical magic measure.


\subsection{Asymptotic WLN structure in even/odd cat codes}
To visualize the asymptotic WLN structure derived above, we specialize it to a simple bosonic-code setting. We focus on the \(d=2\) case and use the standard even/odd cat-code basis,
\begin{equation}
\begin{split}
\ket{0}_L&=
\frac{1}{\sqrt{\mathcal N_{0|2}}}
\bigl(\ket{\alpha}+\ket{-\alpha}\bigr),\\
\ket{1}_L
&=
\frac{1}{\sqrt{\mathcal N_{1|2}}}
\bigl(\ket{\alpha}-\ket{-\alpha}\bigr).
\label{eq:evenOddCat}
\end{split}
\end{equation}
It is convenient to investigate the WLN structure of logical states in this code. The logical states \(\ket{0}_L\) and \(\ket{1}_L\) are exactly orthogonal for any finite \(|\alpha|\), which allows the physical WLN landscape of encoded qubit states to be tracked continuously as \(|\alpha|\) varies. Moreover, the two-point constellation \(\{\alpha,-\alpha\}\) satisfies the non-degeneracy condition. Hence, in the limit \(|\alpha|\to\infty\), this cat code realizes the non-degenerate large-separation regime in its simplest form. In this limit, the cat basis is related to the asymptotically orthogonal coherent-state basis \(\{\ket{\alpha},\ket{-\alpha}\}\) by a \(2\times2\) discrete Fourier transform, allowing the asymptotic WLN structure derived above to be visualized directly in the qubit cat-code setting.

We parameterize an arbitrary pure logical qubit state as
\begin{equation}
    \ket{\psi(\theta,\phi)}
=
\cos(\theta/2)\ket{0}_L
+
e^{i\phi}\sin(\theta/2)\ket{1}_L .
\end{equation}
In the large-separation limit, Theorem~\ref{thm:asymptotic_d_cat_WLN} gives
\begin{equation}
    \mathrm{WLN}_\infty(\ket{\psi})
    =
    \log_2\!\left(
        1+\frac{2}{\pi}\sqrt{1-x^2}
    \right),
    \label{eq:qubitCatWLN}
\end{equation}
where $x=\sin\theta\cos\phi $ is the logical Bloch-sphere coordinate. Thus the asymptotic WLN depends only on \(x\), and its level sets are circles of constant \(x\) around the logical \(x\)-axis. The minimum \(\mathrm{WLN}_\infty(\ket{\psi})=0\) is attained at
\begin{equation}
        \ket{\pm}_L=\frac{\ket{0}_L\pm\ket{1}_L}{\sqrt2},
\end{equation}
which asymptotically coincide with the coherent states \(\ket{\pm\alpha}\). The maximum is attained on the great circle \(x=0\), namely the \(y\)--\(z\) equator. Viewed along the \(x\)-axis, the ideal asymptotic landscape therefore appears as concentric circular contours in the projected \(y\)--\(z\) disk, as shown in Fig.~\ref{fig:catWLNx}(d).

Since the even/odd cat basis remains exactly orthogonal for every finite
$|\alpha|$, the exact WLN defines a well-posed physical landscape throughout the
entire convergence process. Figs.~\ref{fig:catWLNx}(a)--(c) show how the finite-$|\alpha|$ landscapes progressively recover the asymptotic structure. We can see that, for small $|\alpha|$, the contours are visibly distorted relative to the ideal
circular form. As $|\alpha|$ increases, these distortions are steadily suppressed: the level sets become increasingly circular, the minimum moves toward the logical states $\ket{\pm}_L$, and the maximal-WLN region approaches the outer boundary of the equator $x=0$. This indicates that the circular contour structure and the closed form in Eq.~\eqref{eq:qubitCatWLN} are strictly asymptotic. At any finite $|\alpha|$, the exact WLN deviates from this ideal geometry even though the logical basis itself remains exactly orthogonal. Quantifying these finite-$|\alpha|$ deviations will later allow us to assess how robustly the asymptotic WLN structure can be realized physically.

For completeness, we also introduce in Appendix~\ref{app:deltaCat} an intriguing $\delta$-cat encoding, obtained by replacing coherent-state wave packets by
$\delta$ distributions. Although this encoding is nonphysical, it reproduces the ideal large-separation WLN structure exactly and therefore serves as a useful reference model for the asymptotic results.

\begin{figure} 
\centering 
\includegraphics[width=1\linewidth]{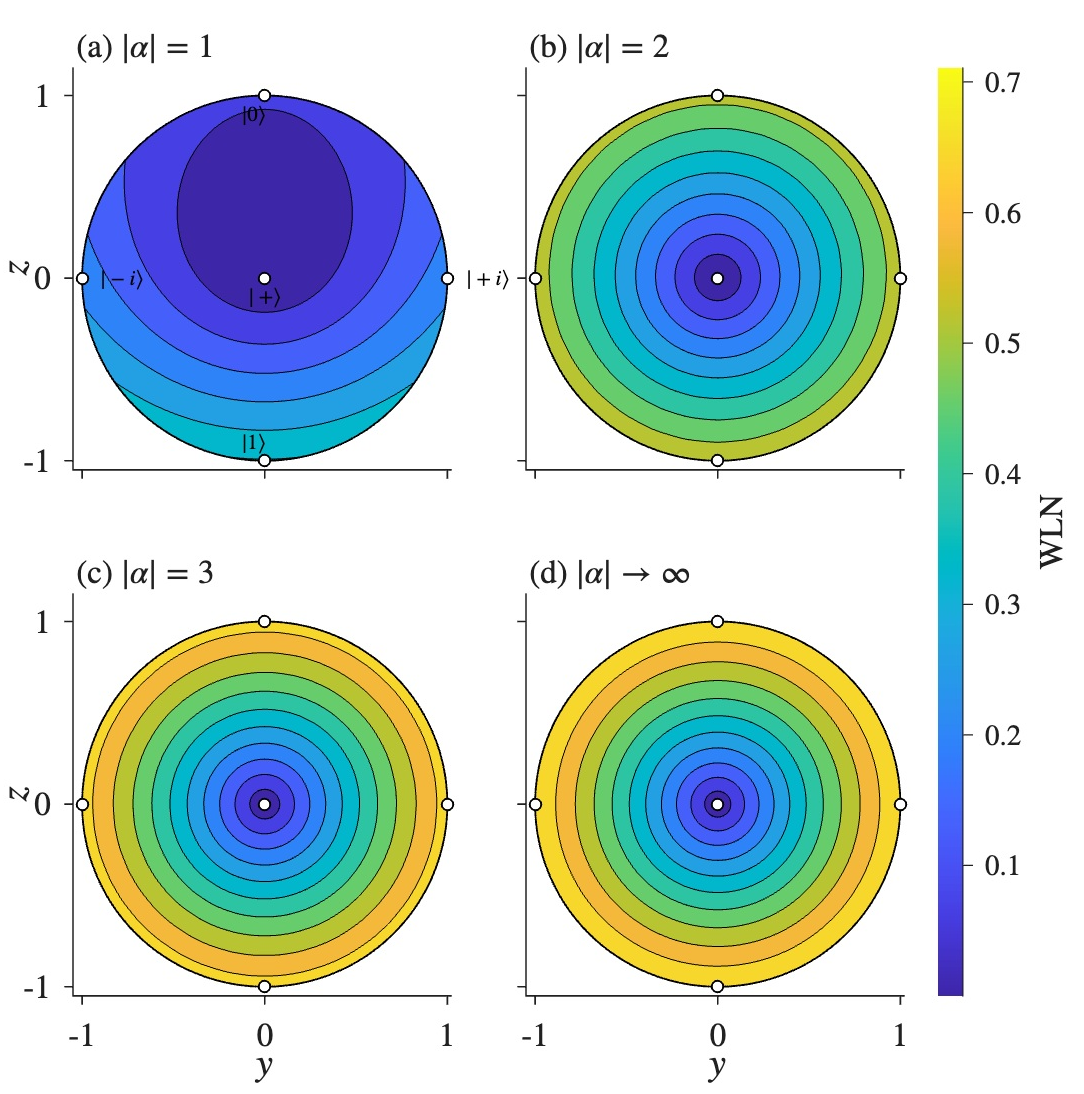} 
\caption{
WLN contours on the logical Bloch sphere for the even/odd cat code, viewed along the positive logical $x$-axis. (d) shows the ideal large-separation limit $|\alpha|\to\infty$, where the equal-WLN contours are concentric circles centered on the $x$-axis; the minimum occurs rigorously at $\ket{\pm}_L$, while the maximum is attained on the outer boundary corresponding to the equator $x=0$. (a)--(c) display the exact WLN landscapes for finite $|\alpha|=1,2,3$. As $|\alpha|$ increases, the WLN landscapes converge toward the asymptotic geometry: the contour distortions diminish, the level sets become more nearly circular, and the minimum shifts toward $\ket{\pm}_L$.
}
\label{fig:catWLNx} 
\end{figure}

\section{From global incompatibility to local SU($d$) compatibility with the magic measure}
\label{sec:compatible}
In the previous section, we showed that the asymptotic WLN of $d$-peaked cat states has a rigid local phase-invariant geometric structure. Since an exact magic--NG correspondence has been established for GKP encodings, we now ask whether a meaningful relation can arise in cat codes by comparing this WLN geometry with the geometry of the magic measure. We first show that, when the magic measure is described through a phase-operator quasiprobability basis, its local structure is naturally organized by sign sectors and their associated eigenbases. This structure, locally determined by the sector-eigenbasis, is globally incompatible with the WLN geometry under cat codes. Nevertheless, both geometries possess the same type of local phase-rotation invariance, which provides the route to constructing the SU($d$) asymptotic cat-code compatibility.

\subsection{Sector-eigenbasis structure of the magic measure and global incompatibility with cat-code WLN}
Throughout this section, we use a qudit magic measure $\mathcal M$ defined as the $\ell_1$ norm of the quasiprobability distribution associated with a phase-operator basis
\(\mathcal O_d=\{O_{l,m}\}_{l,m=0}^{d-1}\)~\cite{Hahn2025}. This measure provides a computable quantifier of the magic resource for arbitrary-dimensional pure states; it reduces to mana in odd dimensions and recovers the stabilizer R\'enyi entropy in the multi-qubit case. The definition and the properties used below are summarized in Appendix~\ref{app:magic}.

Since $\mathcal M$ involves absolute values of quasiprobability components, the state space is naturally decomposed into sign sectors. For a fixed sign pattern
\(\epsilon_{l,m}\in\{\pm1\},l,m=0,...d-1\), we define the corresponding sector by
\begin{equation}
S_{\epsilon}
:=
\left\{
\ket{\psi}\ \middle|\
\operatorname{sgn}\!\left[
\operatorname{Tr}\!\left(
O_{l,m}\ket{\psi}\bra{\psi}
\right)
\right]
=
\epsilon_{l,m},
\ \forall\, l,m
\right\}.
\end{equation}
Inside a fixed sector, the absolute values in $\mathcal M$ are resolved. The measure is therefore locally determined by the expectation value of the Hermitian signature operator
\begin{equation}
A_\epsilon:=\sum_{l,m}\epsilon_{l,m}O_{l,m}.
\end{equation}
Let \(v_{\max}\) be a normalized eigenvector of \(A_\epsilon\) associated with its largest eigenvalue, and complete it to an orthonormal eigenbasis
\begin{equation}
    U_\epsilon=(v_{\max},v_1,\dots,v_{d-1})\in U(d).
    \label{eq:u_epsilon}
\end{equation}
This basis is adapted to the local structure of the magic measure in the chosen sign sector.
\begin{lemma}[Local phase invariance of the magic measure in a sign sector]
\label{lemma:magic}
Fix a sign sector \(S_\epsilon\), and let $ U_\epsilon=(v_{\max},v_1,\dots,v_{d-1})$
be an eigenbasis of the corresponding signature operator \(A_\epsilon\). For any state
\(\ket{\psi}\) lying in this local sector, write
\begin{equation}
    \ket{\psi}
    =
    a_{\max} v_{\max}
    +
    \sum_{j=1}^{d-1}a_j v_j,
    \quad a_j\in\mathbb{C},
\end{equation}
with \(a_{\max}\) chosen real and non-negative without loss of generality. Then, within the sign sector \(S_\epsilon\), the magic measure \(\mathcal M(\ket{\psi})\) depends only on the amplitudes \(|a_j|\) and is independent of the relative phases of the coefficients.
\end{lemma}

Thus, near the local maximum of the magic measure $\mathcal{M}$, $v_{\max}$, the $\mathcal{M}$-landscape is organized by a $U(1)^{d-1}$ phase-rotation-invariant structure in the sector eigenbasis. This invariance is local to the connected portion of the phase orbit that remains inside $S_\epsilon$; crossing a sign-sector boundary changes the relevant signature operator and therefore the local description of the landscape. This can be understood as the direct analogue, on the magic measure side, of the phase-invariant structure found for the asymptotic WLN of cat states. The proof of Lemma~\ref{lemma:magic} is given in Appendix~\ref{app:magic}.

\paragraph*{Qubit example}--
We first illustrate the sector-eigenbasis structure in the qubit case, where it has a simple Bloch-sphere interpretation. The phase-operator basis in the qubit case can be written as
\begin{equation}
    \mathcal O_2=\{I,\sigma_X,-\sigma_Y,\sigma_Z\},
\end{equation}
which is equivalent to the Pauli basis up to a sign convention. Consider the sign sector \(\epsilon=(1,1,1,1)\), which corresponds to positive expectation values of all operators in \(\mathcal O_2\). The associated signature operator is
\begin{equation}
    A_\epsilon
    =
    I+\sigma_X-\sigma_Y+\sigma_Z
    =
    \begin{pmatrix}
        2 & 1+i\\[4pt]
        1-i & 0
    \end{pmatrix}.
\end{equation}
Its eigenvalues are $\lambda_\pm=1\pm\sqrt3,$ with normalized eigenvectors
\begin{equation}
\begin{split}
    \ket{v_+}
    &=
    \sqrt{\frac{3+\sqrt3}{6}}\ket{0}
    +
    e^{-i\pi/4}
    \sqrt{\frac{3-\sqrt3}{6}}\ket{1},
    \\
    \ket{v_-}
    &=
    -
    e^{-i\pi/4}
    \sqrt{\frac{3-\sqrt3}{6}}\ket{0}
    +
    \sqrt{\frac{3+\sqrt3}{6}}\ket{1}.
    \label{eq:qubitbasis}
\end{split}
\end{equation}

According to Lemma~\ref{lemma:magic}, the eigenvector \(\ket{v_+}\) defines the local $\mathcal{M}$-maximizer in this sign sector, while \(\ket{v_-}\) gives the orthogonal transverse direction. Therefore, for any nearby state in the sector written as $\ket{\psi}=\cos\frac{\eta}{2}\ket{v_+}+e^{i\gamma}\sin\frac{\eta}{2}\ket{v_-}$, within the fixed sign sector, the magic measure $\mathcal{M}$ depends on the weight parameter \(\eta\), but is invariant under the relative phase \(\gamma\). The equal-$\mathcal{M}$ contours are therefore generated by phase rotations around the axis defined by \(\ket{v_+}\).

On the Bloch sphere, \(\ket{v_+}\) corresponds to the direction 
\begin{equation}
    (r_x,r_y,r_z)=(1,-1,1)/\sqrt3,
    \label{eq:axis}
\end{equation}
while \(\ket{v_-}\) is the antipodal state. This indicates that the local level sets of $\mathcal{M}$ are circular contours centered around the axis
$r_x=-r_y=r_z$ in this sector. Repeating the same construction for the other sign sectors leads to the full symmetric qubit magic-measure landscape. This sector-eigenbasis description is consistent with the stabilizer R\'enyi entropy 
\begin{equation}
    \mathcal M(\rho)=(1+|r_x|+|r_y|+|r_z|)/2.
\end{equation}
As shown in Fig.~\ref{fig:qubitMagic}(a), the magic measure is distributed symmetrically over the eight octants, each corresponding to a different sign sector. Within each sector, the contours form circles centered around the corresponding sector eigenaxis. In particular, for the sector \(\epsilon=(1,1,1,1)\), which corresponds to the octant \((+,-,+)\), the distinguished axis is Eq.~\eqref{eq:axis}. Viewing the Bloch sphere along this direction, as in Fig.~\ref{fig:qubitMagic}(b), makes the local circular level-set structure explicit.

\paragraph*{Qutrit example}--
As a higher-dimensional example, we consider the qutrit case. We choose the sign sector in which only the component associated with \(O_{0,0}\) is non-positive, while all other components are non-negative. The corresponding signature operator is
\begin{equation}
    A_\epsilon
    =
    \sum_{l,m=0}^{2}\epsilon_{l,m}O_{l,m}
    =
    \begin{pmatrix}
        1 & 1+i\sqrt3 & 0\\
        1-i\sqrt3 & 3 & -2\\
        0 & -2 & 1
    \end{pmatrix}.
\end{equation}
Its eigenvalues are $\{5,1,-1\}$. Let $\{v_{\max},v_1,v_2\}$ be the  orthonormal eigenbasis, with $v_{\max}$ associated with the largest eigenvalue,
\begin{equation}
\begin{split}
    &\ket{v_{\max}}=\frac{1}{\sqrt6}\ket{0}+\sqrt{\frac23}e^{-i\pi/3}\ket{1}+\frac{1}{\sqrt6}e^{i2\pi/3}\ket{2},\\
    &\ket{v_1}=-\frac{1}{\sqrt2}\ket{0}+\frac{1}{\sqrt2}e^{i2\pi/3}\ket{2},\\
    &\ket{v_2}=\frac{1}{\sqrt3}e^{-i\pi/3}\ket{0}+\frac{1}{\sqrt3}e^{i\pi/3}\ket{1}+\frac{1}{\sqrt3}e^{i\pi/3}\ket{2}.
    \end{split}
\end{equation}
The eigenvector $\ket{v_{\max}}$ defines the local $\mathcal{M}$-maximizer in this sector. A general nearby state can be parameterized as $\ket{\psi}=\cos\frac{\eta_1}{2}\ket{v_{\max}}+\\
    e^{i\gamma_1}\sin\frac{\eta_1}{2}\cos\frac{\eta_2}{2}\ket{v_1}+e^{i\gamma_2}\sin\frac{\eta_1}{2}\sin\frac{\eta_2}{2}\ket{v_2}.$ 
By Lemma~\ref{lemma:magic}, the magic measure depends on the amplitude parameters \(\eta_1,\eta_2\), but not on the phase variables \(\gamma_1,\gamma_2\). Thus the qutrit magic-measure landscape again exhibits a local \(U(1)^2\) phase-rotation invariance in the sector eigenbasis.

\paragraph*{Global incompatibility in cat encoding}--
The sector-eigenbasis structure of the magic measure makes the mismatch with the cat-code WLN geometry explicit. In the large-separation limit, the asymptotic WLN level sets are circles organized around the coherent-state axis, whereas the magic-measure landscape is organized by sign-sector-dependent axes. At finite coherent amplitude, Fig.~\ref{fig:catWLNx} shows that the qubit WLN contours are further distorted away from the ideal circular asymptotic geometry, providing additional evidence of the geometric mismatch. We summarize the asymptotic qubit result as follows.

\begin{theorem}[Global level-set incompatibility in the asymptotic qubit cat encoding]
\label{thm:qubit_global_incompatibility}
In the asymptotic even/odd qubit cat encoding, the WLN landscape and the phase-operator magic landscape do not admit a global level-set-preserving correspondence on the logical Bloch sphere.
\end{theorem}
\begin{proof}
Eq.~\eqref{eq:qubitCatWLN} shows that, in the asymptotic even/odd qubit cat code, \(\mathrm{WLN}_\infty\) depends only on the Bloch coordinate \(r_x\). Its level sets are therefore circles around the logical \(x\)-axis, with zero-WLN states at the coherent-state directions. By contrast, the qubit phase-operator magic measure $\mathcal{M}$ is organized by sign sectors; in each sector, the local level sets are centered around a sector-dependent eigenaxis satisfying \(|r_x|=|r_y|=|r_z|\). Hence the WLN and the $\mathcal{M}$-landscapes define different level-set decompositions of the logical Bloch sphere. Therefore, no global level-set-preserving correspondence exists between the asymptotic WLN landscape and the phase-operator magic landscape in the standard even/odd cat encoding.
\end{proof}

The structural mismatch extends to qudit cat encodings. In dimension \(d\), the phase-operator
magic measure partitions the logical state space into \(\mathcal O(2^{d^2})\) sign sectors, each with its own eigenbasis \(U_\epsilon\) and local phase-rotation-invariant structure. In contrast, asymptotic cat-state WLN has a global coherent-basis structure: the WLN depends only on \(|c_i|\) in the coherent-component basis and is invariant under all relative phases. Thus, the magic-measure geometry is organized locally by sign sectors, whereas asymptotic WLN is organized globally by coherent-component amplitudes. This mismatch shows that \textit{the standard cat-code WLN landscape does not support a global, value-preserving GKP-type identification with the phase-operator magic measure. The form of any magic--non-Gaussianity relation is therefore encoding-dependent.} However, the phase-rotation invariance shared by the two geometries provides the mechanism for the local alignment constructed below.


\subsection{Local compatibility between WLN and the magic measure via asymptotic ACC}
Although the WLN and magic-measure landscapes are not globally compatible, their shared local \(U(1)^{d-1}\) phase-rotation invariance suggests a natural local alignment. We now show that this alignment can be achieved by choosing the asymptotic ACC basis according to the sector eigenbasis \(U_\epsilon\).

Once the sector eigenbasis is fixed, we identify the $\mathcal{M}$-maximizer with one coherent component and the remaining eigenvectors with the transverse coherent components,
\begin{equation}
    \begin{split}
        &v_{\max}\leftrightarrow \ket{\alpha_0},\\
    & v_j\leftrightarrow \ket{\alpha_j},
    \quad j=1,\dots,d-1,
    \end{split}
\end{equation}
where the coherent-state constellation satisfies the non-degenerate large-separation condition that $|\alpha_i-\alpha_j|\gg 1$ for any $i\neq j$. Up to local phase choices, this eigenbasis defines an SU($d$) matrix, denoted by $\tilde U_\epsilon$. If $A_\epsilon$ has degenerate eigenspaces, the construction is not unique; any orthonormal basis within the degenerate subspace gives a valid local ACC, with the alignment understood up to this residual unitary freedom. Acting with the corresponding basis transformation on the asymptotically orthogonal coherent-state components gives the SU(\(d\)) asymptotic cat code
\begin{equation}
    |\mu\rangle_{\epsilon,L}^{(\infty)}
    :=
    \sum_{i=0}^{d-1}
    (\widetilde U_\epsilon)_{\mu i}\ket{\alpha_i},
    \;
    \mu=0,\dots,d-1 .
    \label{eq:asymptotic_SUd_cat_code}
\end{equation}
With this choice, the WLN distinguished direction is mapped to the $\mathcal{M}$-maximizing direction \(v_{\max}\), while the transverse phase directions are mapped to those of the sector eigenbasis. Hence the local WLN level sets and the local $\mathcal{M}$-level sets share the same \(U(1)^{d-1}\) phase-rotation geometry, where the sign sector \(\epsilon\) specifies the local region of logical state space to be aligned and \(\widetilde U_\epsilon\) specifies the corresponding cat-code basis transformation.

We illustrate the construction in the qubit case. For the sign sector
\(\epsilon=(1,1,1,1)\) of
\(\mathcal O_2=\{I,\sigma_X,-\sigma_Y,\sigma_Z\}\), the sector eigenbasis is
\(\{\ket{v_+},\ket{v_-}\}\), as given in Eq.~\eqref{eq:qubitbasis}. As shown in Fig.~\ref{fig:qubitMagic}(c) and (d), this SU($2$) ACC aligns the level sets of WLN with those of the magic measure in the target sign sector $\epsilon$. In this sector, the WLN minimizer is mapped to the local $\mathcal{M}$-maximizer $(r_x,r_y,r_z)=(1,-1,1)/\sqrt{3}$, and the nearby WLN contours reproduce the same circular level-set structure around the sector eigenaxis. Thus the alignment is not value-preserving: it matches the local level-set geometry, rather than identifying high WLN with large $\mathcal{M}$.

It is worth noting that, after a sector is chosen, the distinguished SU($d$) ACC gives every encoded state the same type of WLN phase-rotation invariance, but the actual level-set alignment remains sector-specific because the eigenbasis depends on the sign sector. Thus the construction aligns WLN with the magic-measure geometry in the target sector used to define \(U_\epsilon\). In the qubit example, the opposite sector \((-,+,-)\) is also aligned, as shown in Fig.~\ref{fig:qubitMagic}(c), because its eigenbasis coincides with that of \((+,-,+)\) up to eigenvector reordering. Such additional sector alignment is not expected in general dimensions.

\begin{figure}
    \centering
    \includegraphics[width=1\linewidth]{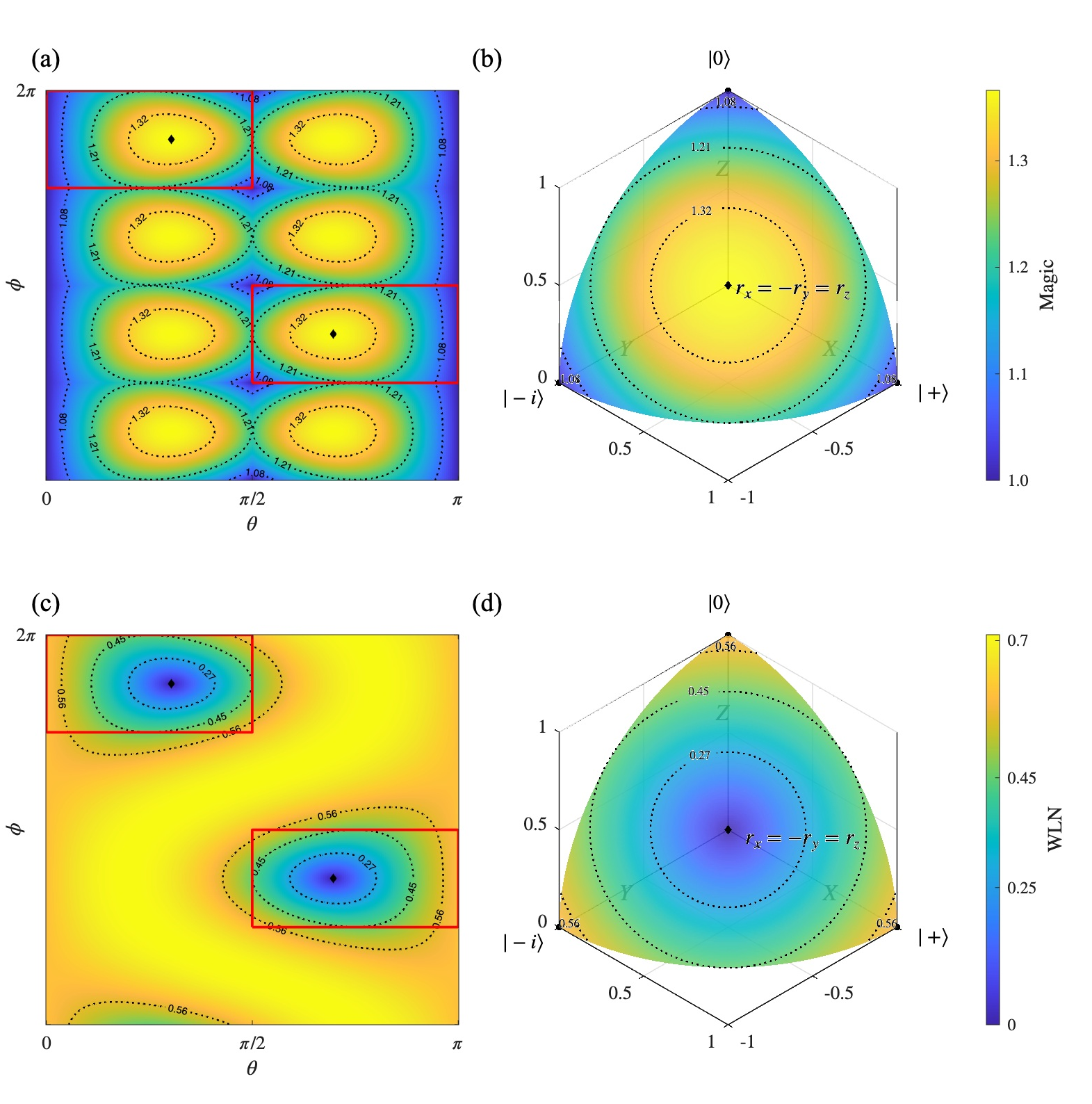}
    \caption{
Qubit illustration of the local level-set alignment between the magic measure and asymptotic WLN under the SU($2$) ACC.
(a) Magic-measure landscape on the flattened Bloch sphere. The two opposite sign sectors with coordinate signatures $(+,-,+)$ and $(-,+,-)$ are highlighted. The full landscape is symmetric over the eight octants, with one local maximum and circular level sets in each sign sector.
(b) Local view of the sector $(+,-,+)$. The local maximum of the magic measure lies along the sector eigenaxis \(r_x=-r_y=r_z\), and the equal-$\mathcal{M}$ contours form concentric circles around this axis.
(c) Asymptotic WLN landscape after applying the SU($2$) ACC. The WLN landscape is divided into two regions centered at the two coherent-state minima, but it shares a global phase-rotation-invariant structure in the rotated cat-code basis. The highlighted regions indicate the two sign sectors where the WLN level-set geometry locally aligns with the magic-measure level-set geometry via the SU(2) ACC.
(d) Local view of the sector \((+,-,+)\). The WLN minimum is aligned with the $\mathcal{M}$-maximum, and the nearby WLN contours form concentric circles around the same axis \(r_x=-r_y=r_z\), reproducing the local circular level-set geometry of the magic measure.
}
    \label{fig:qubitMagic}
\end{figure}

\section{Benchmarking physical realizations of the magic--WLN alignment}\label{sec:realizations}
The previous sections established that, in the large-separation regime, one can construct an ACC such that the level sets of the physical WLN are locally aligned with those of the logical magic resource. In particular, in the qubit case, the level sets are circles organized around a distinguished axis on the Bloch sphere. Recall that Refs.~\cite{Feng2024,Hahn2025} use the ideal GKP encoding for an exact connection between the logical magic resource and physical non-Gaussianity quantified by WLN. That connection is global and quantitative, and it indeed admits the local level-set alignment we introduced here. However, in both settings, physically realizable encodings necessarily involve finite parameters. A natural question is therefore how robust the ideal geometric alignment remains when the ideal codewords are replaced by finite-energy physical states. To address this, we benchmark, in the qubit case, how accurately the two physical encodings reproduce their respective ideal circular level-set geometries.

For the physical cat encoding, we use the SU($2$) cat code in Eq.~\eqref{eq:asymptotic_SUd_cat_code}, with finite coherent amplitude $\alpha_{1,2}=\pm\alpha$. For the physical GKP encoding, we consider its standard finite-energy approximation~\cite{Glancy2006,Matsuura2020,GKP2001}
\begin{equation}
\begin{split}
    \psi_0(q)
    &:=
    \langle q|0\rangle_L
    =
    N_0\sum_{a\in\mathbb Z}
    c_a\,g_\sigma(q-2a\sqrt{\pi}),
    \\
    \psi_1(q)
    &:=
    \langle q|1\rangle_L
    =
    N_1\sum_{b\in\mathbb Z}
    c_b\,g_\sigma(q-(2b+1)\sqrt{\pi}),
    \label{eq:approxGKP}
\end{split}
\end{equation}
where
\begin{equation}
\begin{split}
    g_{\sigma}(q-q_0)
    &=
    \frac{1}{(\pi\sigma^2)^{1/4}}
    \exp\!\left[-\frac{(q-q_0)^2}{2\sigma^2}\right],\\
    c_s
    &=
    \exp\!\left[
        -\frac{(2s\sqrt{\pi})^2}{2\Delta^2}
    \right].
    \label{eq:coeffGKP}
\end{split}
\end{equation}
Under this parametrization, the ideal GKP limit is recovered as the spikes become sharply localized $(\sigma\rightarrow 0)$ and the envelope broadens ($\Delta\rightarrow\infty$). The relevant control parameter is the ratio $\Delta/\sigma$ for quantifying the distance to the ideal GKP code. Details are given in Appendix~\ref{app:approxGKP}. 

To benchmark the physical realization of a code for recovering level-set alignment, we compare geometric deviations from the ideal circular level sets. Specifically, for a physical encoding, we sample $N$ logical states with equal WLN, and express them in an orthonormal basis of the plane perpendicular to the ideal rotation axis, with coordinates denoted by $\{(x_n,y_n)\}_{n=1}^{N}$. Ideally, they should be located on a circle.
To quantify the geometric error, we compute their radial distances $\{r_n|r_n=\sqrt{x_n^2+y_n^2}\}_{n=1}^{N}$, and calculate their standard deviation from the ideal circular radius to obtain the error with this WLN value. We then average the error over different WLN values, and use the averaged error, $\mathcal E$, as the geometric error of this physical code. Intuitively, it measures how accurately the WLN level sets reproduce the ideal circular level set, since in both the ideal GKP and asymptotic cat cases, the corresponding circular level sets lead to $\mathcal E=0$. Here, we compare $\mathcal E_{\mathrm{GKP}}(\Delta/\sigma)$, the geometric error of the approximate GKP code with parameter $\Delta/\sigma$, with $\mathcal E_{\mathrm{cat}}(|\alpha|)$ for a finite cat code. This comparison should be interpreted with respect to the local contour-deviation metric $\mathcal E$; it is not a general comparison of the physical resources required for cat-state and GKP-state preparation.

We numerically evaluate how $\mathcal E_{\mathrm{cat}}(|\alpha|)$ and $\mathcal E_{\mathrm{GKP}}(\Delta/\sigma)$ vary with their respective physical parameters. The results are shown in Fig.~\ref{fig:realization}. Our numerics show that, in both cases, the geometric error decreases as the relevant physical parameter increases. When comparing the performance of the two codes, for the tolerance $\mathcal E_{1}=1\times10^{-2}$, the finite SU($2$) cat code reaches this threshold at $|\alpha|^{\ast}(\mathcal E_{1})\geq 2.04$. By contrast, the approximate GKP code requires $(\Delta/\sigma)^{\ast}(\mathcal E_{1})\geq 20.7$ in order to reproduce the ideal local geometry within $\mathcal E_{1}$. For a smaller tolerance, such as $\mathcal E_{2}=5\times10^{-3}$, the corresponding thresholds are $|\alpha|^{\ast}(\mathcal E_{2})\geq 2.2$ for the cat code and $(\Delta/\sigma)^{\ast}(\mathcal E_{2})\geq 50$ for the approximate GKP code.

A rough comparison can be drawn from existing experimental demonstrations of approximate GKP states~\cite{Fluhmann2019,deNeeve2022,Matsos2024,Matsos2025,Larsen2025} and cat states~\cite{ourjoumtsev2006generating,gerrits2010generation,takase2021generation,endo2023non,endo2025high}. The preparation of approximate GKP states requires highly structured phase-space comb states, stronger squeezing, and significantly more complex state-engineering protocols~\cite{pizzimenti2024optical}. Recent results report $\Delta/\sigma \le 17$ with a squeezing level $\le 10$ dB, which still carries a geometric error of $1.21\times 10^{-2}$ as seen in Fig.~\ref{fig:realization}. In contrast,
approximate cat states can be generated using low-order non-Gaussian operations on Gaussian resources, typically requiring only squeezed vacuum states, linear optics, and a few heralded photon subtractions/additions. Cat states of $|\alpha| >2$ with high fidelity~\cite{takase2021generation} can be generated, for which the corresponding geometric error, $1.13\times 10^{-2}$, is also lower than the GKP value within this metric. Our observations, together with the available evidence, suggest that the local geometric alignment examined here is more readily accessible with current cat-state capabilities, and may therefore offer a practically relevant route for investigating magic--non-Gaussianity relations beyond ideal GKP encodings.

\begin{figure}
    \centering
    \includegraphics[width=1\linewidth]{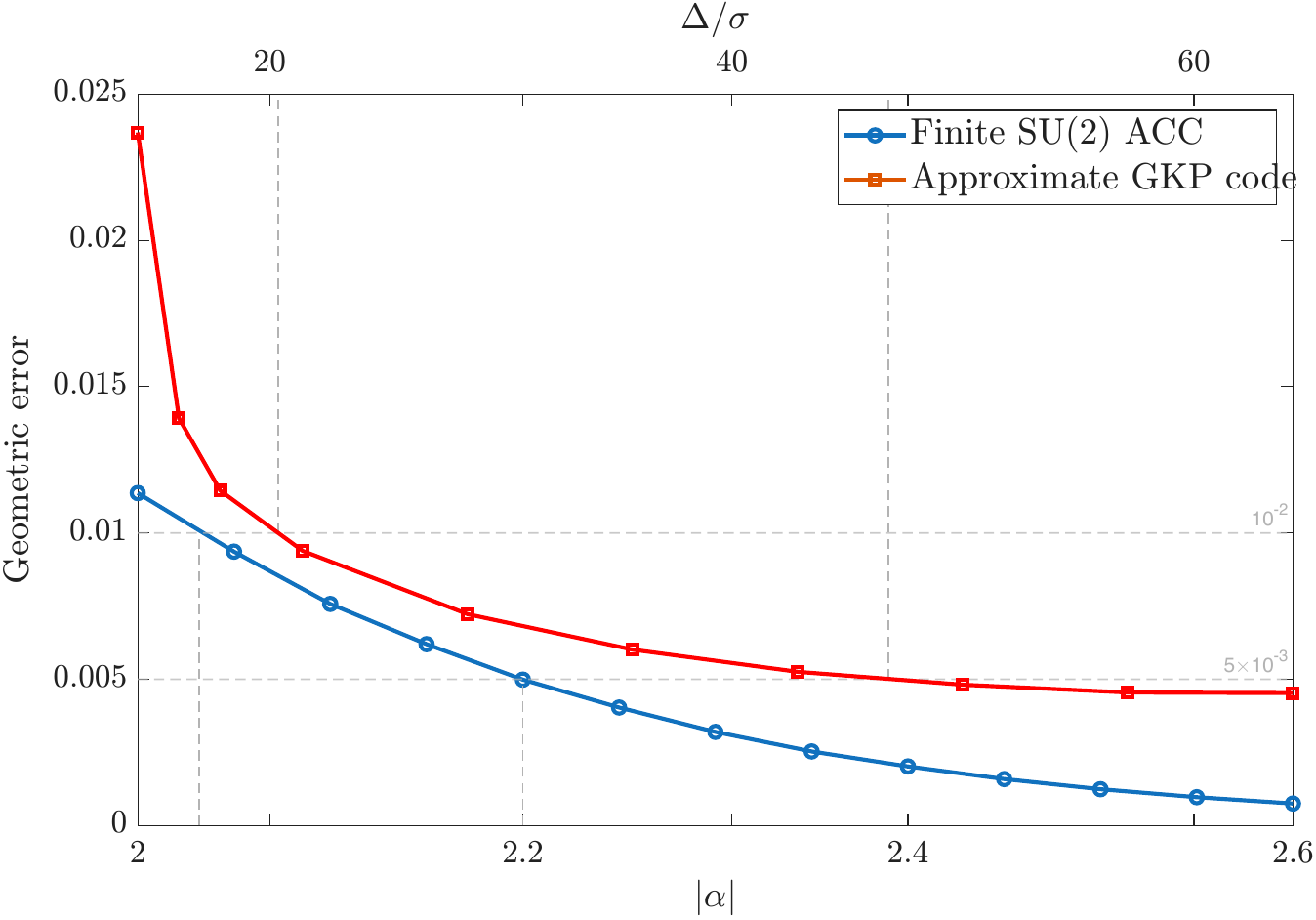}
\caption{
Geometric error for the approximate qubit GKP encoding and the finite SU($2$) cat encoding, plotted as functions of the control parameters $\Delta/\sigma$ and $|\alpha|$. The error \(\mathcal E\) measures the radial deviation of equal-WLN contours from the ideal circular level sets in the target Bloch-sphere octant. We extract \(80\) WLN contours from an \(80\times80\) logical-state grid. For the GKP data, we use a truncated finite-energy model with fixed spike width \(\sigma_0\simeq0.28\), and varying envelope width \(\Delta\). In both encodings, \(\mathcal E\) decreases as the corresponding physical parameter increases. More numerical details can be found in Appendix~\ref{app:numerics_benchmark}.
}
    \label{fig:realization}
\end{figure}

\section{Conclusion and Discussion}\label{sec:conclusions}
In this work, we establish a local geometric connection between the qudit magic measure and continuous-variable non-Gaussianity through an SU($d$) asymptotic cat code. Starting from the WLN of $d$-peaked cat states, we identify a rigid local phase-rotation-invariant structure in the large-separation regime. Exploiting this structure, we show that the cat-code WLN landscape does not reproduce the global, value-preserving magic--NG correspondence known for GKP encodings. Instead, by choosing a sector-adapted SU($d$) asymptotic cat-code basis, one can construct a local level-set alignment between WLN and the magic measure.

Our work addresses the question of whether a meaningful connection between magic and non-Gaussian resources can persist beyond the GKP setting. We show that such a local connection exists in cat encodings and can be established constructively from the intrinsic local geometry of the two resources. In this sense, our result is complementary to the exact GKP correspondence in Refs.~\cite{Feng2024,Hahn2025}. This suggests that cat-code architectures may provide a useful finite-parameter setting for exploring local geometric features reminiscent of GKP magic--NG correspondences, without requiring the full GKP lattice structure.
We further benchmark the physical realization of this alignment against the approximate qubit GKP family considered here, and find that the cat-code construction reaches the same contour-deviation tolerance at moderate cat amplitudes.
Despite the exact magic--NG correspondence available in the ideal GKP encoding, approaching this correspondence with physical finite-energy GKP states remains experimentally demanding, as discussed in Sec.~\ref{sec:realizations}. The ACC construction considered here offers a complementary route: rather than reproducing the full lattice structure of the GKP encoding, it uses the geometry of cat-code WLN to realize a local level-set compatibility between magic and non-Gaussian resources. 


\begin{acknowledgments}
We thank Hon Wai Lau, Nicolo Lo Piparo, and Peizhe Li for insightful discussions. This work was supported by the Center of Innovation for Sustainable Quantum AI under JST Grant No. JPMJPF2221.
\end{acknowledgments}

\nocite{*}
\bibliographystyle{apsrev4-2}
\bibliography{references}

\newpage
\onecolumngrid
\appendix

\section{High-frequency averaging of $|\cos|$-function}\label{app:limit}
In this section, we prove Eq.~\eqref{eq:limit}, which yields the analytical WLN results in the large-separation limit.
\begin{lemma}
    Let $f \in L^1(\mathbb{R}^2)$ and $\phi \in \mathbb{R}$. 
Let $\{k_n\} \subset \mathbb{R}^2$ be a sequence such that $|k_n| \to \infty$. 
Then
\begin{equation}
    \lim_{n\to\infty} 
\iint_{\mathbb{R}^2} 
f(r)\, \left|\cos(k_n\cdot r + \phi)\right| \, dr
=
\frac{2}{\pi}
\iint_{\mathbb{R}^2} f(r)\,d r.
\end{equation}
\end{lemma}

\begin{proof}
The Fourier expansion of $|\cos t|$ is 
\begin{equation}
    |\cos t|
=
\frac{2}{\pi}
+
\frac{4}{\pi}\sum_{m=1}^{\infty}
c_m
\cos(2mt),
\end{equation}
with
\begin{equation}
    c_m =\frac{(-1)^m}{1-4m^2}.
\end{equation}
The coefficients are bounded by $\sum_{m=1}^\infty|c_m| < \infty.$ Substituting the $\cos t$ expansion into the integral leads to 
\begin{align}
    I_n &:= \iint_{\mathbb{R}^2} f(r)\, |\cos(k_n \cdot r + \phi)| \, d{r}\\
    &=\frac{2}{\pi} \iint_{\mathbb{R}^2} f({r})\,dr\\
    &+\frac{4}{\pi}\sum_{m=1}^{\infty} c_m\iint_{\mathbb{R}^2} f({r})\,\cos\!\big(2m({k}_n \cdot {r} + \phi)\big)\,d{r},
\end{align}
Since $f \in L^1(\mathbb{R}^2)$, its Fourier transform is continuous and satisfies
\begin{equation}
\widehat f(\xi)
=
\iint_{\mathbb{R}^2} f(r)\,e^{-i\xi\cdot r}\,d r\to 0
\quad \text{as } |\xi| \to \infty.
\end{equation}
This implies that, for each fixed $m \ge 1$,
\begin{align}
    &\iint_{\mathbb{R}^2} f(r)\cos\!\big(2m(k_n \cdot r + \phi)\big)\,d r
\\=&
\mathrm{Re}
\left(
e^{i2m\phi}
\iint_{\mathbb{R}^2} f(r)\,e^{i(2m k_n)\cdot r}\,dr
\right)\to 0,\;|2m k_n| \to \infty
\end{align}

For all $n,m$, we have
\begin{equation}
    \left|
\iint_{\mathbb{R}^2} f(r)\cos(\cdots)\,dr
\right|
\le
\iint_{\mathbb{R}^2} |f(r)|\,dr
= \|f\|_{L^1}.
\end{equation}
Since $\sum |c_m| < \infty$, the series is absolutely summable and uniformly bounded.
Therefore, by dominated convergence,
\begin{equation}
    \sum_{m=1}^{\infty} c_m
\iint_{\mathbb{R}^2} f(r)\cos(\cdots)\,dr
\to 0,\;|2m k_n| \to \infty.
\end{equation}
Thus
\begin{equation}
    \lim_{n\to\infty} I_n
=
\frac{2}{\pi}
\iint_{\mathbb{R}^2} f(r)\,dr
\end{equation}
\end{proof}

\section{The $\delta$-cat code as a limiting case of the asymptotic SU($2$) cat code}
\label{app:deltaCat}
In this appendix, we introduce an extreme reference model, which we call the
$\delta$-cat encoding. Its purpose is to make the large-separation limit fully
explicit: the WLN structure of the even/odd cat code converges to this model as
$|\alpha|\to\infty$.

We define the logical qubit basis states as superpositions of position
eigenstates localized at $\pm x_0$,
\begin{align}
    \ket{0}_L
    &=\frac{1}{\sqrt{2}}
    \bigl(\ket{x=x_0}+\ket{x=-x_0}\bigr),\\
    \ket{1}_L
    &=\frac{1}{\sqrt{2}}
    \bigl(\ket{x=x_0}-\ket{x=-x_0}\bigr),
\end{align}
where $\langle{x|x=x_0}\rangle=\delta(x-x_0)$. For a general logical qubit state $\ket{\psi}
    =
    \cos(\theta/2)\,\ket{0}_L
    +
    e^{i\phi}\sin(\theta/2)\,\ket{1}_L$, the corresponding Wigner function can be written as
\begin{equation}
    W_{\ket{\psi}}(x,p)
    =
    \frac{1}{2\pi}
    \Big[
        c_+\,\delta(x-x_0)
        +
        c_-\,\delta(x+x_0)
        +
        f(p)\,\delta(x)
    \Big],
\end{equation}
with
\begin{equation}
    \begin{split}
        c_+
    &=\frac{1+\sin\theta\cos\phi}{2},
    \qquad
    c_-=\frac{1-\sin\theta\cos\phi}{2},\\
    f(p)
    &=\cos\theta\cos(2x_0p)
      +\sin\theta\sin\phi\sin(2x_0p).
    \end{split}
\end{equation}
Thus, the Wigner function of the logical state is supported on three vertical
lines in phase space, namely $x=\pm x_0$ and $x=0$. On the two outer lines,
$x=\pm x_0$, the weights are the constants $c_\pm$, whereas on the central line
$x=0$ the Wigner function becomes the trigonometric function $f(p)$. In
particular, this interference contribution is periodic along the $p$ direction,
with period
\begin{equation}
    T=\frac{\pi}{x_0}.
\end{equation}
Since the Wigner function is periodic in $p$, it is natural to quantify the
WLN over one fundamental period, namely $p\in[0,\pi/x_0)$. We therefore define
the normalization over one period by
\begin{equation}
    \mathcal N_T(\ket{\psi})
    :=
    \int_{\mathbb R}dx
    \int_{0}^{\pi/x_0}dp\,
    W_{\ket{\psi}}(x,p)
    =
    \frac{1}{2x_0},
\end{equation}
and the corresponding absolute phase-space volume by
\begin{align}
    V_T(\ket{\psi})
    &:=
    \int_{\mathbb R}dx
    \int_{0}^{\pi/x_0}dp\,
    \bigl|W_{\ket{\psi}}(x,p)\bigr| \\
    &=
    \frac{1}{2x_0}
    \left(
        1+\frac{2}{\pi}\sqrt{1-\sin^2\theta\cos^2\phi}
    \right).
\end{align}
The period-normalized WLN is then
\begin{equation}
    \mathrm{WLN}_T(\ket{\psi})
    :=
    \log_2\!\left(
        \frac{V_T(\ket{\psi})}{\mathcal N_T(\ket{\psi})}
    \right),
\end{equation}
which yields
\begin{equation}
    \mathrm{WLN}_T(\ket{\psi})
    =
    \log_2\!\left[
        1+\frac{2}{\pi}\sqrt{1-\sin^2\theta\cos^2\phi}
    \right].
\end{equation}
This is exactly Eq.~\eqref{eq:qubitCatWLN}. Therefore, the $\delta$-cat code
reproduces the ideal asymptotic WLN structure of the even/odd cat code exactly,
and can be viewed as its limiting reference model as $|\alpha|\to\infty$.

\section{Local structure of the magic measure and proof of Lemma~\ref{lemma:magic}}
\label{app:magic}
In this appendix, we summarize the magic measure used in the main text and prove Lemma~\ref{lemma:magic}. The magic measure we use is defined on a phase-operator basis. It is computable in arbitrary dimension, reduces to the discrete-Wigner-based magic measure in odd dimensions, and recovers the stabilizer R\'enyi entropy in the multiqubit case~\cite{Hahn2025}.

Let $\rho$ be a single-qudit state of dimension $d$. We consider the Hermitian
operator basis
\begin{equation}
    O_{l,m}=e^{-i\pi ml/d}\,M_l Z_d^m,
\end{equation}
where $l,m\in\mathbb Z_d$,
\begin{align}
    M_l=\sum_{\substack{u,v\in\mathbb Z_d\\u+v=l\;(\mathrm{mod}\;d)}}\ket{u}\bra{v},
    \;\\
    Z_d=\sum_{j=0}^{d-1}e^{2\pi i j/d}\ket{j}\bra{j}.
\end{align}
These operators form a Hermitian basis of the qudit operator space and can be
viewed as a Hermitian generalization of the Pauli basis to arbitrary dimension. Strictly speaking, the operators are naturally defined for indices in $\mathbb Z_{2d}$. In the present work, however, we
restrict to $l,m\in\mathbb Z_d$, since the additional indices only reproduce the same operators up to sign factors, which are absorbed into the signature pattern. We define the corresponding coefficients by
\begin{equation}
    \chi_\rho(l,m):=\frac{1}{d}\operatorname{Tr}\!\left(O_{l,m}\rho\right),
\end{equation}
so that the magic measure is
\begin{equation}
    \mathcal M(\rho)=\sum_{l,m\in\mathbb Z_d}\bigl|\chi_\rho(l,m)\bigr|.
    \label{appeq:pureMagic}
\end{equation}
Thus, $\mathcal M(\rho)$ is the $\ell_1$ norm of the operator-basis coefficients of $\rho$.
After defining the sign sector, the structure of $\mathcal{M}$ in the logical space can be described sector by sector. In each sector, we have the following lemma.
\begin{lemma}
Let \(\ket{\psi}\in S_\epsilon\) be expanded in the eigenbasis of \(A_\epsilon\) as
\begin{equation}
    \ket{\psi}
    =
    a_{\max} v_{\max}
    +
    \sum_{j=1}^{d-1}a_j v_j,
    \qquad a_j\in\mathbb C,
\end{equation}
with \(a_{\max}\) chosen real without loss of generality. Then, within the fixed sign sector \(S_\epsilon\), the magic measure \(\mathcal M(\ket{\psi})\) depends only on the amplitudes \(|a_j|\) and is independent of the relative phases of the coefficients.
\end{lemma}

\begin{proof}
For any state $\ket{\psi}$ in the neighborhood of $v_{\max}$ in the same sign sector, $\ket{\psi}$ can be written as
    \begin{equation}
    \ket{\psi}
    =
    \sqrt{1-\sum_{j=1}^{d-1}|a_j|^2}\,v_{\max}
    +\sum_{j=1}^{d-1} a_j v_j,
    \qquad a_j\in\mathbb{C}.
\end{equation}
Since $\ket{\psi}$ also lies in the  $\epsilon$ sector, substituting it into the expression for $\mathcal{M}$ in Eq.~\eqref{appeq:pureMagic} gives
\begin{align}
    d\cdot\mathcal{M}(\ket{\psi})
    &=
    \left(1-\sum_{j=1}^{d-1}|a_j|^2\right)\lambda_{\max}
    +\sum_{j=1}^{d-1}|a_j|^2\lambda_j \notag\\
    &=
    \lambda_{\max}
    -\sum_{j=1}^{d-1}(\lambda_{\max}-\lambda_j)|a_j|^2.
\end{align}
Hence, near the extremum $v_{\max}$, the local magic measure depends only on the amplitudes $|a_j|$ and is completely independent of the phases of $a_j$. 
\end{proof}

\section{Approximate GKP code and its Wigner functions}\label{app:approxGKP}
To characterize the WLN behavior of physical states under the approximate GKP code defined in Eq.~\eqref{eq:approxGKP}, we calculate their Wigner functions in this appendix. The approximate logical states are defined by two parameters, \(\Delta\), the width of the Gaussian envelope, and \(\sigma\), the width of each Gaussian spike. Using the notation of Eq.~\eqref{eq:coeffGKP} in the main text, we denote the spike positions by
\[
q_{0,a}=2a\sqrt{\pi},\qquad q_{1,b}=(2b+1)\sqrt{\pi}.
\]
The normalization factors are given by
\begin{align}
N_0^{-2}
&=
\sum_{a,b} c_a c_b
\exp\!\left[
-\frac{(q_{0,a}-q_{0,b})^2}{4\sigma^2}
\right],
\\
N_1^{-2}
&=
\sum_{a,b} c_a c_b
\exp\!\left[
-\frac{(q_{1,a}-q_{1,b})^2}{4\sigma^2}
\right].
\end{align}
The overlap between the two approximate logical basis states is
\begin{equation}
\langle 0|1\rangle_L
=
N_0 N_1
\sum_{a,b} c_a c_b
\exp\!\left[
-\frac{(q_{0,a}-q_{1,b})^2}{4\sigma^2}
\right].
\label{eq:logical_overlap}
\end{equation}

Since the Wigner function is linear, the Wigner function of an arbitrary approximate logical qubit can be decomposed into four contributions,
\begin{equation}
W_{00}(q,p),\qquad W_{11}(q,p),\qquad W_{01}(q,p),\qquad W_{10}(q,p),
\end{equation}
where \(W_{ij}(q,p)\) denotes the Wigner function of the operator \(\ket{i}_L\bra{j}_L\). Each contribution consists of a sum of Gaussian dyads of the form \(\ket{g_{q_a}}\bra{g_{q_b}}\), whose Wigner function is
\begin{equation}
W_{\ket{g_{q_a}}\bra{g_{q_b}}}(q,p)
=
\frac{1}{\pi}
\exp\!\left[
-\frac{\left(q-\frac{q_a+q_b}{2}\right)^2}{\sigma^2}
-\sigma^2 p^2
\right]
\exp\!\big(i(q_a-q_b)p\big).
\end{equation}
One then obtains
\begin{equation}
\begin{split}
W_{00}(q,p)&=
\frac{N_0^2}{\pi}
\sum_{a,b} c_a c_b\,
\exp\!\left[
-\frac{\left(q-(a+b)\sqrt{\pi}\right)^2}{\sigma^2}
-\sigma^2 p^2
\right]
\cos\!\big(2(a-b)\sqrt{\pi}\,p\big),\\
W_{11}(q,p)&=
\frac{N_1^2}{\pi}
\sum_{a,b} c_a c_b\,
\exp\!\left[
-\frac{\left(q-(a+b+1)\sqrt{\pi}\right)^2}{\sigma^2}
-\sigma^2 p^2
\right]
\cos\!\big(2(a-b)\sqrt{\pi}\,p\big),\\
W_{01}(q,p)&=
\frac{N_0 N_1}{\pi}
\sum_{a,b} c_a c_b\,
\exp\!\left[
-\frac{\left(q-\left(a+b+\frac{1}{2}\right)\sqrt{\pi}\right)^2}{\sigma^2}
-\sigma^2 p^2
\right]
\exp\!\big(i(2a-2b-1)\sqrt{\pi}\,p\big).
\end{split}
\end{equation}
By definition, \(W_{10}(q,p)=W_{01}(q,p)^*\). Therefore, for an arbitrary logical qubit
\[
\ket{\psi}=\alpha\ket{0}_L+\beta\ket{1}_L,
\]
the normalized Wigner function takes the form
\begin{equation}
W_{\ket{\psi}}(q,p)
=
\frac{
|\alpha|^2 W_{00}(q,p)
+
|\beta|^2 W_{11}(q,p)
+
\alpha\beta^* W_{01}(q,p)
+
\alpha^*\beta W_{01}(q,p)^*
}{
|\alpha|^2
+
|\beta|^2
+
2\operatorname{Re}\!\big(
\alpha^*\beta \langle 0 | 1 \rangle_L
\big)
}.
\end{equation}

\subsection{Numerical details for the finite-parameter benchmark}
\label{app:numerics_benchmark}

Here we summarize the numerical choices used in the finite-parameter comparison between the approximate GKP encoding and the finite SU($2$) cat encoding. The benchmark quantity is the geometric error \(\mathcal E\), defined from the radial deviation of equal-WLN contours from the corresponding ideal circular level sets in the target Bloch-sphere octant.

For the approximate GKP encoding, we use a truncated finite-energy comb of Gaussian spikes. The spike width is fixed at
\begin{equation}
    \sigma_0\simeq 0.28,
\end{equation}
while the envelope width \(\Delta\) is varied. This choice satisfies $6\sigma_0 \leq \sqrt{\pi},$ so that the effective \(3\sigma_0\) support of each Gaussian spike does not significantly overlap with neighboring spikes separated by \(\sqrt{\pi}\). For each value of \(\Delta\), the comb is truncated to include approximately $6\Delta/(2\sqrt{\pi})$ Gaussian peaks, corresponding to the effective \(3\Delta\) support of the Gaussian envelope. 

The phase-space grid is chosen to resolve both the individual Gaussian spikes and the interference fringes. In the \(q\) direction, the grid spacing is set to $dq=\sigma_0/5$ to resolve each spike width. In the \(p\) direction, if \(k_{\max}\) denotes the largest relevant position separation in the truncated comb, we use $dp=\pi/(6k_{\max})$ which gives approximately six points per shortest interference period. These choices ensure that both the Gaussian envelope and the fastest interference oscillations are resolved.

For the finite SU($2$) cat code, the Wigner function is evaluated on a fixed phase-space window $q,p\in[-5,5]$ using an \(800\times800\) grid. In the parameter range used below, this window contains the relevant support of the coherent components and their interference fringes. Increasing the grid resolution or enlarging the phase-space window does not produce a visible change in the plotted geometric error at the scale shown. The final numerical values used in Fig.~\ref{fig:realization} are listed in Tables~\ref{tab:cat_error_data} and~\ref{tab:gkp_error_data}.

\begin{table}[h]
\centering
\begin{tabular}{c|c}
\hline
Cat amplitude \( |\alpha| \) & Geometric error \( \mathcal E_{\mathrm{cat}} \) \\
\hline
2.00 & \(1.14\times10^{-2}\) \\
2.05 & \(9.36\times10^{-3}\) \\
2.10 & \(7.57\times10^{-3}\) \\
2.15 & \(6.20\times10^{-3}\) \\
2.20 & \(4.99\times10^{-3}\) \\
2.25 & \(4.03\times10^{-3}\) \\
2.30 & \(3.20\times10^{-3}\) \\
2.35 & \(2.53\times10^{-3}\) \\
2.40 & \(2.02\times10^{-3}\) \\
2.45 & \(1.59\times10^{-3}\) \\
2.50 & \(1.24\times10^{-3}\) \\
2.55 & \(9.68\times10^{-4}\) \\
2.60 & \(7.51\times10^{-4}\) \\
\hline
\end{tabular}
\caption{
Geometric error for the finite SU($2$) cat encoding as a function of the cat amplitude \( |\alpha| \).
}
\label{tab:cat_error_data}
\end{table}

\begin{table}[h]
\centering
\begin{tabular}{c|c|c}
\hline
Envelope width \(\Delta\) & Scale ratio \(\Delta/\sigma_0\) & Geometric error \( \mathcal E_{\mathrm{GKP}} \) \\
\hline
4.00  & 14.29 & \(2.37\times10^{-2}\) \\
4.50  & 16.07 & \(1.39\times10^{-2}\) \\
5.00  & 17.86 & \(1.14\times10^{-2}\) \\
6.00  & 21.43 & \(9.39\times10^{-3}\) \\
8.00  & 28.57 & \(7.22\times10^{-3}\) \\
10.00 & 35.71 & \(6.01\times10^{-3}\) \\
12.00 & 42.86 & \(5.25\times10^{-3}\) \\
14.00 & 50.00 & \(4.81\times10^{-3}\) \\
16.00 & 57.14 & \(4.55\times10^{-3}\) \\
18.00 & 64.29 & \(4.53\times10^{-3}\) \\
\hline
\end{tabular}
\caption{
Geometric error for the approximate qubit GKP encoding as a function of the scale-separation ratio \(\Delta/\sigma_0\), with fixed spike width \(\sigma_0\simeq0.28\).
}
\label{tab:gkp_error_data}
\end{table}
\end{document}